\documentclass[journal,11pt,draftcls,onecolumn]{IEEEtran}
\usepackage{amsmath,amssymb,amsbsy,amsthm,geometry}

\let\epsilon\varepsilon
\let\phi\varphi

\let\epsilon\varepsilon

\newtheorem*{lemma*}{Lemma}
\newtheorem{theorem}{Theorem}

\newtheorem{claim}{Claim}

\begin{document}


\title{ Linear hash-functions and their applications to error detection and correction} 

\author{
 \IEEEauthorblockN{Boris Ryabko\\}
 \IEEEauthorblockA{Federal Research Center for Information and Computational Technologies and  
Novosibirsk state university, \\
Novosibirsk, Russian Federation, Email: boris@ryabko.net\\}}
\date{}

\maketitle


\begin{abstract}
We describe and explore so-called linear hash functions and show how they can be used to build error detection and correction codes. The method can be applied for different types of errors (for example, burst errors). When the method is applied to a model where number of distorted letters is limited, 
   the obtained  estimate  of its performance is slightly better than the known Varshamov-Gilbert  bound. We also describe random code whose performance is close to the same boundary, but its construction is much simpler. 
In some cases the obtained methods are simpler and more flexible than the known ones. In particular, the complexity of the obtained error detection code and the well-known CRC code is close, but the proposed code, unlike  CRC, can  detect  with certainty errors whose number does not exceed a predetermined limit. 


\end{abstract}


\section{Introduction} 
Error detection and correction codes are commonly used in telecommunication and data storage systems, and there are   many effective and practically used constructions of such codes,  see for review \cite{lc,vp,kk}. Currently, cyclic redundancy check (CRC), which was proposed in
\cite{pet}, is one of the most popular error detection codes, while block codes \cite{vp} are the basis of many error correction methods. 
 
 In short, error correction and detection systems can be described as follows: a binary word $ x_1 ... x_L $ is transmitted through a communication channel, and the recipient receives a message $y_1 ... y_L $ in which some letters $ y_i $ may differ from $ x_i $.
The purpose of an error detection code is to inform the receiver that some  letters sent were changed during the transmission (i.e., at least one $ x_i \neq y_i $). The purpose of an error correction codes is not only to report that the  errors have occurred, but also to find all the  letters that were changed (that is, all $ i $ for which $ x_i \neq y_i $). (We consider the most popular model in which messages are words in the binary alphabet $\{0,1\}$, but the main results can be easily extended to any finite alphabet.) 

The main part of both types of  codes can be described as follows: the transmitted word $ x_1 ... x_L $ contains two subwords, say $ x_1 ... x_{L-l} $  and $ x_{L-l+1} ... x_L $, $L > l \ge 1$,  where the first subword contains information bits, and the second one contains so-called  check bits (or parity bits).  When the sender wants to send $ L-l $ bits, he first sets them to $ x_1 ... x_ {L-l} $, and then calculates the check  bits $ x_ {L-l + 1} ... x_L $. The receiver receives the word $ y_1 ... y_L $ and uses it to detect or correct errors that may have occurred during  transmission. 
Generally speaking, check bits are given by a 
 function  $ \lambda $, which is defined on the set of $ (L-l) $ -bit words  with values in the set of $ l $ -bit  words.
In the area of error detection codes, $ \lambda $  is often called a hash function.
It is worth noting that sometimes the check bits are not at the end of the message, but in other places.

The simplest example of this scheme is the parity-bit, or check-bit, method.
In this method, a sequence of information bits is $ x_1 ... x_{L-1} $, the check bit is $x_L$,
(i.e. $l=1$). If the total number of 1-bits in the string  
$ x_1 ... x_{L-l} $ is even, then $x_L = 0$, otherwise  $x_L = 1$. When the receiver obtains $y_1 ... y_L$ he calculates the total number of 1-bits. If this value is odd, it means that an error has occurred.  Thus, this method makes it possible to detect one error, but, obviously, does not detect two errors (and any even number of errors).

Naturally, the larger the number of information bits (i.e. $L-l$), the better the code one can construct. 
That is why the question about codes with the largest number of information bits has attracted attention of many researcher (see for review \cite{vp}).
In order to describe some known results in this field we need some definitions.
The expression $|X|$ denotes  the number of elements if $X$ is a set and the length $X$, if $X$ is a word. 
 Let $u$, $v$  be  finite binary words of the same length. We denote the  Hamming weight  of $u$, i.e., the number of 1's in the word $u$ by  $||u||$ and, by definition, the Hamming distance $d_h(u,v)$ $=|| u\oplus v ||$,  where $\oplus$ is bitwise 
 XOR (or addition modulo 2).    
 Let  $U$ be a set of binary words of the same length and $|U| >1$. The  minimal Hamming distance of $U$ is defined by as $d_h(U) = \min_{u,v \in U, u \ne v} d_h(u,v)$. 
 Let $U$ be a set of binary words of some length $L$, $L \ge 2$.  The Varshamov-Gilbert bound states that 
\begin{equation}\label{vg}
\max_{d_h(U) = d}  |U| \ge \, \, 2^{L- \lceil \log_2 ( 1+\sum_{i=0}^{d-2}  (^{L-1}_{\, \, \,\, i} ) \, \rceil},
\end{equation}
see \cite{vp}, Theorem 2.9.3. 
There exist some improvements of this bound, but  they 
 do not change its asymptotic  (see  for review \cite{imprVG,imprVG2}).

The ability of a code to detect and correct errors is simply related to the Hamming distance $ d_h(U) $.  
To show this, we first define 
\begin{equation}\label{ball}
B_n^m \subset \{0,1\}^m, n \le m, \text{is a set
of words of length m which contain n or less 1's }  .
\end{equation}
(That is, $B_n^m $ contains all words whose   Hamming weight   is not grater than $n$.)
Now, take $U \subset \{0,1\}^L$ 
and  consider a method where  $U$ is the set of  messages  transmitted and $v$ is the word of  errors occurred, 
  that is, the  message  transmitted is $x\in U$ and  the  message received is $y = x\oplus v$.
Suppose that $d_h(U) = d$,  $d \ge 2$.   It turns out, that   $d-1$  errors  can be detected is        (i.e., $v \in B^L_{d-1}$), 
 and  $\lfloor (d-1)/2 \rfloor $  errors    can be corrected  (i.e., $v \in B^L_{\lfloor (d-1)/2 \rfloor }$).  Indeed, if $x \in U$ and $v \in B^L_{d-1}$, then $y=x \oplus v$ does not belong to $U$ and this indicates an error.  In order to correct errors,   the word closest to $ y $ is considered sent. 

We briefly reviewed a model in which errors are letter distortions, and their number is limited by a  certain bound. There are other models of possible errors that describe various systems for transmitting and storing information, for example, packet errors.
This general case is also  considered in this work, the part 4.

In this work we describe new classes of error detection and correction codes, which are based on the so-called linear hash functions.
Linear hash functions are defined as follows:   any map  $\lambda $   defined on $ L $-bit binary words whose values $ \lambda (x) $ are taken from the set $ l $-bit binary
 words, $ l< L$, is called a hash  function. 
  (Formally, $\lambda : \{0,1\}^L  \to \{0,1\}^l$, $ L > l \ge 1 $.) 
A hash function $\lambda$ is called linear if for any $L$-bit words $x$ and $y$ 
$$ \lambda(x \oplus y) = \lambda(x) \oplus \lambda(y)  .$$
Linear hash functions are well-known and  date back at least to Zobrist  \cite{zo}.

The proposed  methods allow  us to build a code for any set of errors (including the case when errors occur in packages). In particular, this method can be used to detect errors whose number does not 
 exceed a predetermined limit (for example, detecting any three errors).  Note that the well-known cyclic redundancy check (CRC) codes do not detect a predetermined number of errors with for certainty; 
 rather,  CRC make it possible to detect a predetermined number of errors (say, 3) only with a certain probability.

It is worth noting, that the performance of the proposed codes slightly exceeds the well-known  
Varshamov-Gilbert (VG) bound \cite{vp}.

When considering error correction  and detection codes, the problem of the complexity of the method is very important.  Three questions arise: the complexity of i) encoding,  ii) decoding, and iii)  constructing encoding and decoding methods.  
In the case of error-detecting the encoding and decoding is quite simple, whereas the complexity of constructing encoding and decoding methods  is relatively large. To overcome this, we propose a randomized algorithm for constructing an encoder and decoder whose performance is close to optimal, 
 but the complexity is much smaller. 
 
         The rest of the paper is organised as follows. The next section contains a description of some of the properties of linear hash functions, as well as a general scheme of their application to codes.
 Section~\ref{s:codes}  
   is devoted to a model in which  errors are letter distortions and an upper bound on their number is given. 
    First, we describe a code which meets the VG bound. This method is then generalized in two directions: we describe its modification that performs slightly better than the VG estimate, and we propose a randomized algorithm.   The last section describes general methods of  error detection  and correction.

\section{Linear hash functions and their applications to error detection and correction}
\subsection{Representation of linear hash functions as sums of words} 
Consider a linear hash function $ \lambda $ defined on the set of $ L $-bit binary words $ \{0,1 \}^L $ and $ \lambda (x) $ taken from the set of $ l $ -bit binary words $ \{0,1 \} ^l $. It will be convenient to denote by $ e^k_i $ a string of $ k $ -bits that contains 1 at the $i$-th position and zeros at all  others, and let $ e^k_0 $ be the string of length $k$ consisting only of 0s. 

Let  $x = x_1 ... x_L$ be an $L$-bit word and $v_1, ..., v_L$ be any $l$-bit words.
Define a function 
\begin{equation}\label{linear}
\lambda(x) = x_1\times 	 v_1 \oplus   x_2 \times 	v_2  \oplus ...\oplus  x_L \times v_L,
\end{equation}
where $x_i \in \{0,1\}$ and we assume $ 0 \times v =0 0 ... 0$, $1\times v = v$. 

For any two vectors $x, y$ we obtain 
$$
\lambda(x \oplus y) = ( x_1\oplus y_1) \times 	 v_1 \oplus   (x_2 \oplus y_2)  \times 	v_2  \oplus ...\oplus ( x_L \oplus y_L) \times v_L \, \, =
$$
$$
  ( x_1 \times  v_1) \oplus (  y_1 \times 	 v_1 ) \oplus   ( x_2 \times  v_2) \oplus (  y_2 \times 	 v_2 ) \oplus ...\oplus ( x_L \times  v_L) \oplus (  y_L \times 	 v_L )
  \, = \lambda(x) \oplus \lambda(y) \, . 
$$
So, the hash-function (\ref{linear}) is linear. 
On the other hand,  for any linear hash-function $\lambda'$ 
$$ \lambda' (x) = x_1 \times \lambda'(e^L_1) \oplus x_2 \times \lambda'(e^L_2) \oplus
 ... \oplus x_L \times \lambda'(e^L_L) 
 $$
 and, hence, $\lambda'$ is represented in the form (\ref{linear}), 
where $v_i =  \lambda'(e^L_i) $. Thus, we derived the following: 
 \begin{theorem}\label{general-h-claim}
A hash-function $\lambda$ is linear if and only if it can be represented as
\begin{equation}\label{t1}
\lambda(x) = x_1\times 	 v_1 \oplus   x_2 \times 	v_2  \oplus ...\oplus  x_L \times v_L  
\end{equation}
 for some $l$-bit words $v_1, ... , v_L \in  \{0,1\}^l \, .$
\end{theorem}
 Note that  the CRC code is a linear has function and, hence, can be represented as (\ref{linear}).
Also, it is worth noting that the calculation of (\ref {linear}) does not require multiplication or other time-consuming operations, and can be    performed in linear time.

\subsection{A scheme for using a linear hash function to detect and correct errors}

Consider the following data transfer scheme: there are sets of $ L $ -bit messages $ A_0 $ and possible distortions (or errors) $ D \subset \{0,1\}^L $.
If the message $ x \in A_0 $ is sent through  the channel, 
a distortion $ d \in D $ may  occur, that is, the recipient receives the message $ x \oplus d $. 
(For example, if $ D $ contains all words with two 1's, this means that two-bit errors may occur during the transfer.)   

A key component is a linear hash function  $ \lambda $ such that  
\begin{equation}\label{key}
 \lambda(x) = e^l_0 \text{ for  all } x \in A_0
\text{  and } \lambda(d) \neq e^l_0 \text{ for  all } d \in D,
\end{equation}
where the set 
  $ A_0 $ is constructed as follows: any message $ x = x_1 ... x_L $ consists of $L-l$  information symbols $x_{i_1} ... $ 
$x_{i_{L-l}}$, while the remaining $l$ symbols are used as check symbols. 
 (Generally, we will use
 $ x_1 ... x_{L-l} $  as  information symbols and $ x_ {L-l + 1} ... x_L $ as  check symbols.) 
The check symbols are chosen in such a way     that $ \lambda (x) = e^l_0 $ for all $x \in A_0$.
If the distortion $d \in D$ occurs, the received message $y$ can be presented as
$y = x \oplus d$. (If no  error  occurs, then $y = x$.)
 We can see from this equation that this method gives a possibility to detect any  distortion $d \in D$, because 
\begin{equation}\label{main2}
\lambda(x) = e^l_0 \, , \, \, 
\lambda(y) = \lambda(x\oplus d) = \lambda(x) \oplus  \lambda(d) = \lambda(d) \neq e^l_0 \, .
\end{equation}
Thus, this scheme allows  to detect any distortion $d \in D$, because the equation $ \lambda (y) \neq 0 $ means that   $ d $ occurred, and, conversely, the opposite equation $ \lambda (y) = 0 $ informs about the absence of an error.

This system can be used to correct errors if the following additional property applies:
all values of $ \lambda (d) $ are different, i.e. for all $ d_i, d_j \in D $,   $\,\,\, \lambda (d_i) \neq \lambda (d_j) $.  Indeed, in this case, the decoder may first
compute $ \lambda (y) = \lambda (d) $, see (\ref{main2}). All $ \lambda (d) $ are different and therefore the decoder can find 
 $ d $ from $ \lambda (d) $ and compute
$ x = y \oplus d $.

\section{Codes for a limited number of letter errors}\label{s:codes}
We  consider codes which can detect  or correct a limited number of bit-errors, that is, the possible distortions belong to the ball  $B^L_d$ of a certain radius  $d$,
$L > d \ge 1$.
For this purpose we develop some methods for constructing such a liner hash-function $\lambda$, $ \lambda: \{0,1\}^L \to \{0,1\}^l$, $l \le L$ 
 and a set $A_0$  that 
\begin{equation}\label{d-prop} 
 d_h(A_0)=d, d \ge 2,  
 \, \, \lambda(x)= e^l_0 \, \,   for \, \,  any \, \,  x \in A_0  \, \,  and \, \, 
\lambda(y) \neq e^l_0 \,\, for \, \, any \, \,  y \in \{0,1\}^L \setminus
 A_0 \,.
\end{equation} 

The following property of this construction  will play an important rule.

 \begin{theorem}\label{prop1}  Let there be a linear hash-function $\lambda$, an integer $d, d \ge 2$, and a set $A_0$ for which $\lambda(x) = e^l_0$ for any $x \in A_0$ and $\lambda(y) \neq e^l_0$ for any 
$y \in \{0,1\}^L   \setminus A_0$. Then $d_h(A_0)$ $\ge d$, $d >1$, if and only if 
$\lambda(v) \neq e^l_0$ for any $v \in B^L_{d-1} \setminus e^l_0$.
\end{theorem}
\begin{proof}
Suppose that $d_h(A_0)$ $=d$.  Then,  for any $x \in A_0$ and any $v \in B^L_{d-1} \setminus e^L_0$,  the word $x\oplus v$ does not belong to $A_0$, 
because $d_h(x, (x\oplus v) )$ $=||v|| \le d-1$.  Hence, $\lambda (x \oplus v) \neq e^l_0$.
 From this we obtain $\lambda(v) = \lambda(x) \oplus \lambda(v) = \lambda(x \oplus v) \neq e^l_0$. 
 
  Let us prove the opposite statement. Suppose, $\lambda(v) \neq e^l_0$ for all $v \in 
B^L_{d-1}  \setminus \, e^L_0 $. Let $x \in A_0$, $v \in 
B^L_{d-1} $. We can see that
$x \oplus v$ does not belong to $A_0$, because 
$\lambda(x \oplus v) = \lambda(x ) $ $ \oplus \lambda(v)$ $= e^l_0 \oplus \lambda(v)$ 
$ \neq e^i_0$. 
  So, if $0 < ||v|| < d$, for some $v$, then  $x\oplus v$ does not belong to $A_0$ and, hence, $d_h(A_0) \ge d$. 
\end{proof}

The construction  (\ref{d-prop}) can be  directly used in error detection and correction codes. Indeed,  as  mentioned in the introduction,
 those codes are as follows: 
   either 

i)  a code that can  detect $d-1$ or less bit-errors, or  

ii) a code that can correct  $\lfloor (d-1)/2 \rfloor$ or less bit-errors.

In accordance with this, we will call the hash function $\lambda$ and the set $ A_0 $ in   (\ref{d-prop}) as a code.
In this section 
we consider three methods for constructing  such codes. 
  The first method produces a code that matches the   VG  bound and can be easily randomized. A slightly improved estimate will be valid for the second method, while the third method is a  greatly simplified version of the first one, obtained using randomization.

\subsection{Method which meets VG bound}

Here our goal is        to build a code (\ref{d-prop}) for given    integers $L$ and $d$, 
$L > d \ge 2$.  
It means that we should  find  methods   i) to  calculate $l$,  ii)  to build $\lambda $ and iii) to describe how to find, for any 
information symbols $x_1 ... x_{L-l}$, the symbols 
 $x_{L-l+1} ... x_L$ for which $\lambda(x_1 ... x_{L})$  $= e^l_0$ (that is,  $x_1 ... x_{L}$ $\in A_0$) .

\subsubsection{Building the hash function $\lambda$}
The following algorithm (Algorithm 1) is intended to find $l$ and $\lambda$ while a method for performing 
 iii) will be described immediately after.
 
 \emph{ The input } is  two integers $L, d$. 
 
\emph{  The output }
\begin{equation}\label{l}
l = \left \lceil    \log   \left(   \sum_{i=0}^{d-2}    \left(^{ L-1}_ {\, \, \,\,i\, \, } \right)  +  1  \right) \right\rceil, 
\end{equation}
 a linear hash function $ \lambda: \{0,1 \}^L \to \{0,1 \}^l $, for which (\ref{d-prop}) holds true, and the set $A_0$ 
 (here and below $\log = \log_2$).
If $ l $ in (\ref {l}) is not defined or $ l \ge L $ then the algorithm stops and answers that the solution does not exist.

\emph{  Algorithm 1.} 

\emph{ First step. } Calculate $l$ in (\ref{l}) and define  
\begin{equation}\label{lyahat} 
  \hat{\lambda}(e^L_1 )  =  e^l_1,   \hat{\lambda}(e^L_2 )  =  e^l_2,  \dots,  
   \hat{\lambda}(e^L_l )  =  e^l_l  .
\end{equation}

\emph{ Second step. } For $i = l+1, l+2, ... , L$ define $\hat{\lambda}(e^L_i)$ as follows:
\begin{equation}\label{Di}
 \hat{\lambda}(e^L_i)  = v_i \text{ where  $v_i$ is any  word  from } (  \{0,1\}^l    \setminus \, 	\hat{\lambda}( B^{i-1}_{d-2})
 ).
 \end{equation}
Here and below $\hat{\lambda}(Z) = \bigcup_{z \in Z} \left\{ \hat{\lambda}(z) \right\}$  
for any set 
$Z $,
and  $ B^{i-1}_{d-2}$ is the set of all words $b_1 b_2 ... b_L$ from  $ B^{L}_{d-2}$  such that   $b_i = b_{i+1} =  ... =  b_L = 0$. 
Note that 
i) $\hat{\lambda}(e^L_j ) $, $j = 1, 2, ... , i-1$ are defined when $\hat{\lambda}( B^{i-1}_{d-2})$ is calculated,  and   ii)   the set $ \{0,1\}^l    \setminus \, 	\hat{\lambda}( B^{i-1}_{d-2})$ is not empty for $i = 1, ..., L$, due to $|\hat{\lambda}( B^{i-1}_{d-2})| \le$
$  \sum_{j=0}^{d-2}    (^{ i-1}_ {\, \, \,\,j\, \, } ) \, \,$ and the definition of 
$l$ in (\ref{l}).

From this definition we can see that 
$ \hat{\lambda}(e^L_i) \neq  \hat{\lambda}(w)$ for any $w \in B^{i-1}_{d-2}$ and, hence,   $\hat{\lambda}(w') \neq e^l_0$  for any $w' \in B^{i}_{d-1}$, 
for $i = l+1, l+2, ... , L$.

 From     (\ref{lyahat})  and  (\ref{Di})  we can see that 
 \begin{equation}\label{unicity7}
 \hat{\lambda}(u) \neq e^l_0 \text{ for every }  u \in  B^L_{d-1}   \setminus  \{ e^L_0 \} . 
 \end{equation}

\emph{ Final step. } 
The goal of this step is to permute the values of the hash function $\lambda$ in such a way that the last values $\lambda(e^L_{L-l +1}), ..., \lambda(e^L_{L})$ will be 
the first $l$ values of 
$ \hat{\lambda}$.  Clearly, this step is not a mandatory procedure, 
 but it simplifies the encoding of the information symbols. 
Note that  any permutation of  coordinates  of  the set $B^m_n$ does not change it, so the following   procedure is correct: 
Define  $\lambda$ using $\hat{\lambda}$ as follows: 
$\lambda(e^L_i) = \hat{\lambda}(e^L_{i+l}) $ 
 for $i = 1, ... , L-l $ and $\lambda(e^L_{L-l +i}) = \hat{\lambda}(e^L_i)
$  for $i = 1, ... , l $. 
Note that 
\begin{equation}\label{checkbits}
\lambda(e^L_{L-l+i} ) = e_i^l , 
\end{equation}
for $i = 1, ... , l $, see (\ref{lyahat}). 
From  (\ref{unicity7}) we obtain  
\begin{equation}\label{mainlam}
 \lambda(u) \neq e^l_0 \text{ for all }  u \in   B^L_{d-1}   \setminus \{ e^L_0 \}  . 
\end{equation}
 \subsubsection{Description of the set $A_0$ or encoding}\label{encoding}\label{A0}
Now we can describe the set $A_0$, that is, the method of encoding of information symbols.
Let $x_1 ... x_{L-l}$ be a set 
 of information symbols and we want to find the check symbols $x_{L-l + 1} ... x_L$.   In order to do it, first, we pad   $x_1 ... x_{L-l}$ with  $l$   zeros at the end  and denote  the obtained string as  $u = x_1 ... x_{L-l} 00... 0$. 
Then calculate $ \lambda(u) = w_1 ... w_l$ and define 
$x_{L-l +1} = w_1,  x_{L-l +2} = w_2, ... , x_L =  w_l . $

Taking into account (\ref{checkbits}) we can see that  
$ \lambda(00...0  x_{L-l} ... x_L) $ 
$=  \lambda(00...0  w_1 ... w_l) $ $= ( w_1 ... w_l )$. From this we obtain
  $\lambda( x_1 ... x_{L}) =  $ $\lambda(x_1 ... x_{L-l} 00...0)$
$\oplus \lambda(00...0  x_{L-l} ... x_L) $ $= ( w_1 ... w_l ) $ $\oplus $ $ ( w_1 ... w_l) $ $ = e^l_0$.
So, for any information symbols  $x_1 ... x_{L-l}$ we find the check symbols  $x_{L-l+1} ... x_{L}$
such that $\lambda( x_1 ... x_{L}) = e^l_0$.

It will be convenient to describe the properties of the described algorithm as follows:

 \begin{theorem}\label{vg}  i) The described algorithm is correct, that is,  $d_h(A_0) \ge d$, 
 
 ii) the following inequality is valid for the number of information symbols $L-l$: 
 \begin{equation}\label{t3}
L-l = L -  \, \, \lceil   \, \, \log \,  ( \,  \sum_{i=0}^{d-2}    (^{ L-1}_ {\, \, \,\,i\, \, } ) \,\, \, + \, 1 \, ) \, \, 
\rceil \, , 
\end{equation}
  \end{theorem}
 
\begin{proof}
 Taking into account (\ref{mainlam}), we can obtain the first statement i) from Theorem 2.  The statement ii) follows from (\ref{l}).
 \end{proof}
 
 \subsubsection{The complexity}
 Now consider the complexity of the proposed method.
   There are the  following three  important characteristics to consider: i)  the  time ($T$) to construct the  hash function using the algorithm described,  ii) the tie encoding ($t_{enc}$)  and decoding ($t_{dec})$ time if the method is used for error detection or correction.
   It is important to note that the hash function must be constructed only once and then  used many times (for different inputs), while encoding and decoding are performed for each input. 
   
 \begin{claim}\label{compl}     
i)  The  time $T$  is proportional to $     \sum_{j=0}^{d-2} $ $   (^{ L -1}_{\,\,\,\,  j} ) \,$.
If  $ L $ grows    to $ \infty $ and $d$ is a constant then $T = O((L \log L)^{d-1})$,
if $L \to \infty$ and $\lim d/L $ equals some $\alpha$, then $T= 2^{ L H(\alpha)}$.

ii) If this algorithm 
 is used for error detection then $t_{enc} =$   $t_{dec}$ $=  O( (d-2) L \log L )$. For error correction $t_{enc} $ is the same, but 
$t_{dec}$ is proportional to $T$ in i).
\end{claim}
(here $H(\alpha ) = -(\alpha \log \alpha + (1-\alpha) \log (1-\alpha ) ) $ is Shannon entropy, see   \cite{co}.) 

\begin{proof}

The proof is based on a direct estimation of number of bit-operations 
and   known estimates of the binomial coefficients, see  \cite{co,fe}.  
 \end{proof}

\subsubsection{Examples}
 We start with  the case  $d=2$, which gives a possibility to detect one error. So, the input of the algorithm is $L \ge 1$ and $d=2$. 
 From (\ref{l}) we see that $l = \log (1+1) = 1. $ From (\ref{lyahat}) we obtain $
 \hat{\lambda}(e^L_1) = 1$.
 Taking into account that $B_0^{i-1}$ is $e_0^{L}$, we can see from (\ref{Di}) that 
 $\hat{\lambda}(e^L_i) = 1$ for all $i$, $1 = 2, ... , L$.  From this and (\ref{checkbits})
 we can see that $\lambda(e^L_i) = 1$ for all $i$, $1 =1, 2, ... , L$. 
 The information symbols are $x_1 ... x_{L-1}$, while the check symbol is $x_L$.
 If this method is applied to error detection, the encoder calculates $x_L = \bigoplus_{i=1}^{L-1} (x_i \times \lambda(e^L_i) )= $
  $  \bigoplus_{i=1}^{L-1} (x_i \times 1)$ $=   \bigoplus_{i=1}^{L-1} x_i  $, while the decoder  calculates $\bigoplus_{i=1}^{L} (x_i \times \lambda(e^L_i) )= $
  $  \bigoplus_{i=1}^{L} (x_i \times 1) =$
  $ \bigoplus_{i=1}^{L}  x_i $. If this sum is 1, then one error occurred, otherwise an error did not occur. Thus, in this case, the code based on linear hash functions coincides with the parity check method.

  The second example is  $d=3$. 
 Now the input of the algorithm is $L $ and $d=3$.  From (\ref{l}) we obtain 
 $l =  \lceil \log ( (L-1) + 1 +1 ) \rceil = \lceil \log ( (L+1 ) \rceil .$
 According to the  first step (see (\ref{lyahat}))
 $  \hat{\lambda}(e^L_1 ) \, = \, e^l_1,  \hat{\lambda}(e^L_2 ) \, = \, e^l_2,  \, ... , , \, 
 \hat{\lambda}(e^L_l ) \, = \, e^l_l \, .$
Having taken into account (\ref{Di}) we can see that different values of 
$ \hat{\lambda}(e^L_l )$, $i = l+1, ... , L$
will be assigned different words from $\{0,1\}^l$  $\setminus \{  e^l_1, ..., e_l^l \} $.
From the final step of the algorithm we can see that 
  $  \lambda(e^L_{L-l+1} ) \, = \, e^l_1,  \lambda(e^L_{L-l+2} ) \, = \, e^l_2,  \, ... ,  \, 
\lambda(e^L_L ) \, = \, e^l_l \, $, while the other values of $\lambda$ are different words 
  from $\{0,1\}^l$  $\setminus \{  e^l_1, ..., e_l^l \} $. 
  The encoding is carried out according to the method \ref{encoding} described above.

Note that encoding and decoding can be implemented in such a way that there is no need to store the values $ \lambda (e ^ L_1), \lambda (e ^ L_2), ..., $ $\lambda (e_L ^ L) $. Indeed, one can select the values $ \lambda (e ^ L_1), \lambda (e ^ L_2), ..., $ $\lambda (e_ {L-l} ^ L) $ in lexicographical order and calculate these values sequentially during encoding and decoding. 
  
  It is interesting that, in fact,  the described method  
  is the  
    well-known Hamming code which can either  detect two errors or correct one \cite{vp}. (Indeed,  the described code can correct one error as follows: if the transmitted (or saved) message is $ y $ and one error occurred, then for some $ i $ $ \lambda (y) = \lambda (e^L_i )$. This means that the error occurred in $ i $ - th position. )

  \subsection{Methods  whose performance outperforms the VG bound}
  
The method proposed here is a modification of the previous one. The only difference is the   choice of the new value $ \hat{\lambda}(e^L_i)$ in (\ref{Di}).  That is why we describe only those parts  of the algorithm that are different, 
i.e the output and the second step. 
It will be convenient to describe the method and the purpose 
 of the modifications together.

First, we describe the main idea of the proposed modification.  From (\ref{Di}) we can see that the value of $ l $ is determined by the size of the set $ \hat{\lambda} (B^{L-1}_{d-2}) $, because it must be less than $ 2^l -1 $. 
Then we use the following obvious inequality 
 $ |\hat{\lambda }(B^{L-1} _ {d-2}) | $ $ \le | B^{L-1} _ {d-2} | $ and the requirement $ | B^{L-1} _ {d-2} | $ $ \le 2^l -1 $ instead of $ | \hat{\lambda }(B^{L-1} _ {d-2}) | $ 
 $ \le  2^l -1 $.
In what follows we build such a hash function $\hat{\lambda }$  that $ | \hat{\lambda }(B^{L-1} _ {d-2}) | $ is less than 
$ | B^{L-1} _ {d-2} | $.  For this purpose we find such subsets $U,V$ from 
$  B^{L-1} _ {d-2} $ that $U \cap V = \varnothing$ and $\hat{\lambda } (U) = 
\hat{\lambda }(V)$.  Taking into account that $| \hat{\lambda }(Z) | \le |Z| $ for any $Z$ and the last equation 
  we can see that   
\begin{equation}\label{UV} |\hat{\lambda }(B^{L-1} _ {d-2}) \, |= 
| \hat{\lambda }(B^{L-1} _ {d-2} \,  \setminus \, U )\, | \, \le  | \, B^{L-1} _ {d-2} \,  \setminus \, U \, | \, = \, | B^{L-1} _ {d-2}  |\, - \, | \, U \, |\, .
\end{equation}
So, if we find such sets $U$ and $V$,  we have the upper bound 
$ |\hat{\lambda }(B^{L-1} _ {d-2}) \, | \le$ $ | B^{L-1} _ {d-2}  |\, - \, | \, U \, |$ 
instead of $ |\hat{\lambda }(B^{L-1} _ {d-2}) \, | \le$ $ | B^{L-1} _ {d-2}  |\,$
and, hence, can reduce the number of the check bits $l$. 

Now we can describe the modified algorithm. As we mentioned, the only difference is the second step and the definition of $l$ which are as follows: 

 \begin{equation}\label{l2}
l = \left\lceil    \log \left(   \sum_{i=0}^{d-2}    \left(^{ L-1}_ {\, \, \,\,i\, \, } \right)  
-  \sum_{s = 1}^{d-3}  \left( \left(_{\, \, \, s}^{d-1} \right) \sum_{ j=1}^{s-1} \left(^{L-d-1}_{\, \, \, \, \, \,  j}\right)  \right) +  1   \right)
\right\rceil, 
\end{equation}
and   

\emph{ Second step. } For $i = l+1, l+2, ... , L$ define $\hat{\lambda}(e^L_i)$ as follows:
\begin{equation}\label{Di-22}
 \hat{\lambda}(e^L_i)  = w_i \text{ where $w_i$  is any word  from } 	\hat{\lambda} ( B^{i-1}_{d-1} )   \setminus \,  \hat{\lambda}	( B^{i-1}_{d-2} 
 ).
 \end{equation} 
 Note that $\sum_{s = 1}^{d-3}  \left( \left(_{\, \, \, s}^{d-1} \right) \sum_{ j=1}^{s-1} \left(^{L-d-1}_{\, \, \, \, \, \,  j}\right)  \right) $ corresponds to $|U|$ in (\ref{UV}).

Now we describe the sets $U$ and $V$.  From (\ref{Di-22}) we can see that for $i=L-1\, \, $, 
$\hat{\lambda}(e^L_{L-1} )  = w_{L-1}  \in \hat{\lambda} ( B^{L-2}_{d-1} )   \setminus \,  \hat{\lambda}	( B^{L-2}_{d-2} )$.  By definition, $ \hat{\lambda}(e^L_{L-1})$ $
= \hat{\lambda} (00...0010) $. 
On the other hand, from (\ref{Di-22})  we   see that there exists $x = x_1 ... x_{L-2} 00$ which contains $d-1$ ones and  $ \hat{\lambda}(x) = w_{L-1} $. Hence, 
$\hat{\lambda} (00...0010) = $ $\hat{\lambda}(x) $. The word $x$ contains $d-1$ ones among $  x_1 ... x_{L-2} $. To simplify the notation we suppose that
$x_1 =1,  ... , x_{d-1} = 1$ whereas the  others $ x_i =0$. (We can do this without loss of generality due to the symmetry of the set  $B^{L-2}_{d-1}   \setminus \,    B^{L-2}_{d-2} \, \, $.) So, 
\begin{equation}\label{Di-2}
 \hat{\lambda}( 00...0010) =    \hat{\lambda} (11 ... 1 00 ... 0) \, , \, where \, \, \, 
 11...1 \, \, \, is \, \, \,  (d-1)\, \, \, ones
 .  \end{equation}
\begin{equation}\label{ZUV}Z = \{ z : z = x y 00,  \, \, \,where \, \, \, |x| = d-1,  \,
|y| = L- d -1, \, 
|| x|| + ||y||  \le d-3, \, \, \, and \, \, \, ||y|| \le ||x|| -1 \}
 \end{equation}
 Now we define the following sets 
\begin{equation} 
$$  $$ U = \{u: u = 00...0010  \oplus z, \, z \in Z \},  \, \, \,
V = \{v: v =11 ... 1 00 ... 0  \oplus z, \, z \in Z\} ,
\end{equation} 
 where, as before in (\ref{Di-2}), 
$ 11...1 $  is   $ (d-1)\, $ ones. 
\begin{claim}\label{cl1}

i) $U \cap V = \varnothing$. 

ii)   $\hat{\lambda}(U) =  \hat{\lambda}(V)$.

iii) $U \subset B^{L-1}_{d-2}, \, V \subset B^{L-1}_{d-2}$.


iv)
\begin{equation}\label{Zsize}
|U|  \, \, \, =  \, \, \,\sum_{s = 1}^{d-3}  \left( \left(_{\, \, \, s}^{d-1} \right) \sum_{ j=1}^{s-1} \left(^{L-d-1}_{\, \, \, \, \, \,  j}\right)  \right) \, .
\end{equation}

\end{claim} 

{\it Corollary. } From i) - iii) we can see that $\hat{\lambda}( B^{L-1}_{d-2}) =$ 
$\hat{\lambda}( B^{L-1}_{d-2})  \setminus U)$.  From this and iv) we obtain 
$|\hat{\lambda}( B^{L-1}_{d-2}) |= $ $|\hat{\lambda}( B^{L-1}_{d-2})  \setminus U)|$
$\le $ $| B^{L-1}_{d-2}) | - |U| =$
  $ \sum_{i=0}^{d-2}    (^{ L-1}_ {\, \, \,\,i\, \, } ) \,\, \, 
- \, \, \, \, \, \,\sum_{s = 1}^{d-3}  \left( \left(_{\, \, \, s}^{d-1} \right) \sum_{ j=1}^{s-1} \left(^{L-d-1}_{\, \, \, \, \, \,  j}\right)  \right)$. Taking into account that $|\hat{\lambda}( B^{L-1}_{d-2}) |$ must be not grater than $2^l -1$, we obtain (\ref{l2}). 
 
 {\it Proof. }
i)  The two last digits of any $u \in U$ are 10, 
whereas the two last digits of any $v \in V$ are 00, see (\ref{ZUV}).  

ii)  $|U| = |V|$ (see (\ref{ZUV}) ) and  for any $u \in U$ there exists 
$v \in V $ such that $\hat{\lambda}(u) =  \hat{\lambda}(v)$. (Indeed, for  any $u$:  $u = (00 ... 010) \oplus z, z \in Z$ . Hence, taking into account 
 linearity  of  $\hat{\lambda}$ and (\ref{Di-2}), for $v = 11...100...00  \oplus z$ we obtain 
$\hat{\lambda}(u) = \hat{\lambda}(v)$.)

iii) From the definition $U$ in (\ref{ZUV}) we can see that $||u|| = ||x|| + ||y||+ 1$. Taking into account that $||x|| + ||y|| \le d-3$, we obtain from the last equation that 
$||u|| \le d-2$, that is, $u \in B^{L-1}_{d-2}$.
Let us consider the set $V$.  From (\ref{ZUV}) we can see that 
$||v|| = (d-1) - ||x|| +||y|| $ 
and $||y|| +1 \le ||x|| $.  From the  latter  two inequalities 
we obtain $||v|| \le d-2$,
that is, $v \in B^{L-1}_{d-2}$.


iv) The equation $|U| = |Z|$  follows from  (\ref{ZUV}). Let now $s = ||x||, j = ||y||$. 
From 
 (\ref{ZUV}) we can see that $||y|| \le ||x|| -1 $, that is, $0\le j \le s-1$.  Taking into account that  $||x|| + ||y|| \le d-3 $ and $ y \ge 0$, we can see that  
 $||x||  \le d-3 $, that is, $0 \le s \le d-3$. Using  common  combinatorial formulas, we obtain  
$$
|U|  \, \, \, =  \, \, \,\sum_{s = 1}^{d-3}  ( \left(_{\, \, \, s}^{d-1} \right) \sum_{ j=1}^{s-1} \left(^{L-d-1}_{\, \, \, \, \, \,  j}\right)  )\, ) \, . 
$$ 
{\it The claim is proven.}

It will be convenient to  summarize  the properties of the algorithm just described as follows:
 \begin{theorem}\label{vg}  For the modified algorithm 2, the following equality is valid for the number of information symbols $L-l$: 
 \begin{equation}\label{t3}
L-l = L -  \, \, \lceil   \, \, \log \,  ( \,  \sum_{i=0}^{d-2}    (^{ L-1}_ {\, \, \,\,i\, \, } ) \,\, \, 
- \, \, \, \, \, \,\sum_{s = 1}^{d-3}  ( \left(_{\, \, \, s}^{d-1} \right) \sum_{ j=1}^{s-1} \left(^{L-d-1}_{\, \, \, \, \, \,  j}\right)  ) + \, 1 \, ) \, \, 
\rceil \, . 
\end{equation}
  \end{theorem}

\subsection{A randomised algorithm whose performance is close to the VG bound.} 

In this part we consider a  randomised   algorithm whose performance is close to the VG bound, but whose complexity is much smaller. 

Let, as before,  the block length be $L$, the  required code distance be $d$ and $l = \lceil   \, \, \log \,  ( \,  \sum_{i=0}^{d-2}    (^{ L-1}_ {\, \, \,\,i\, \, } ) $ $\, + \, 1 \, ) \, 
\rceil $, see (\ref{l}). 
Define $l_\Delta  = l+\Delta $, where $\Delta$ is such an integer that   $L- l_\Delta \ge 1$. 

The only difference between the new randomised algorithm and the algorithm 1 is in the second step (\ref{Di}). In the new algorithm the values $\hat{\lambda}(e_i^L)$,
$i =l_\Delta +1, ... L$, 
 are  chosen randomly from $\{0,1\}^{l_\Delta }$ according to the uniform distribution. 
We call this method  Algorithm 3 or the randomised algorithm.  
 
 Our goal is to estimate the probability of the following events  
   \begin{equation}\label{PI}  \Pi  = \{  \text{ For }\,  i = l_\Delta+1, ... L , \text{ the ( randomly chosen ) word  }  \hat{\lambda}(e_i^L)
 \text{ belongs to } \{0,1\}^{l_\Delta} \setminus \hat{\lambda}(B_{d-2}^{i-1}) \},
\end{equation}   
  see (\ref{Di}).
In turn,   if $\Pi$ occurs then
   this  gives a possibility to build an
 encoding set $A_0$ for which $d_h(A_0) \ge d$.
 Define  
  \begin{equation}\label{pii}
\Pi_i = \{ a  \,\,\,uniformly  \,\,\, chosen \,\,\, word \,\,\, u  \,\,\,belongs \,\,\, to  \,\,\, \{0,1\}^{l_\Delta}  \setminus \hat{\lambda}(B_{d-2}^{i-1}) \, \}\,.
 \end{equation}
 Clearly, 
  $\Pi = \Pi_{{l_\Delta}+1} \cap ... \cap \Pi_{L-1}$ and the following chain of equations is valid
 \begin{equation}\label{sum1}
P(\Pi) = P(\Pi_{l_\Delta+1} \cap ... \cap \Pi_{L-1} ) =  P( \Pi_{l_\Delta+1} ) P(\Pi_{l_\Delta+2}| \Pi_{l_\Delta+1})  P(\Pi_{l_\Delta+3}| \Pi_{l_\Delta+2}\Pi_{l_\Delta+1}) ...  $$ $$ P(\Pi_{L-1}|\Pi_{L-2}...\Pi_{l_\Delta+1}) \ge 
\prod_{i=l_\Delta +1}^L  	\frac{| \{0,1\}^{l_\Delta} \,  \setminus \, \hat{\lambda}( B_{d-2}^{i-1}) \,|}{2^{l_\Delta} } 
$$
$$
\ge \prod_{i=l_\Delta +1}^L  ( 	1 -    |B_{d-2}^{i-1}|   /   2^{l_\Delta} )
\ge \prod_{i=l_\Delta+1}^L  ( 	1 -     \sum_{j=0}^{d-2} (^{i-1}_{\,\,j})   /   2^{l_\Delta} )
\ge 1 -   2^{-l_\Delta} \sum_{i=l_\Delta+1}^L     \sum_{j=0}^{d-2} (^{i-1}_{\,\,j})     \, ,
\end{equation}
where 
$|B^{i-1}_{d-2}| = \sum_{j=0}^{d-2} (^{i-1}_{\,\,j})$. 
Here we used two following inequalities: $|\lambda(Z)| \le |Z|$  for any hash-function $\lambda$ and any set $Z$, and $(1 - a)(1-b) \ge 1 - (a+b)$ for non-negative $a$ and $b$.

 This rather cumbersome expression can be simplified to obtain an asymptotic estimate. Indeed, 
 \begin{equation}\label{sum2}
 \sum_{i=l_\Delta+1}^L     \sum_{j=0}^{d-2} (^{i-1}_{\,\,j})
\le \sum_{i=0}^L     \sum_{j=0}^{d-2} (^{i-1}_{\,\,j}) =      \sum_{j=0}^{d-2} 
\sum_{i=0}^L     (^{i-1}_{\,\,j}) \, ,
\end{equation}
where, by definition,  $ (^{a}_{\,b}) = 0$, if $a < b$ or $b < 0$.
Now we will apply the well-known identity
$$
\sum_{m=0}^n (^{m}_{\,k}) = (^{n+1}_{\,k+1})
$$ 
which is sometimes  called the hockey-stick identity (see, for example, \cite{fe}). 
So, from this and (\ref{sum2}) we obtain 
$$  \sum_{i=l_\Delta+1}^L     \sum_{j=0}^{d-2} (^{i-1}_{\,\,j})
\le    \sum_{j=0}^{d-2} 
\sum_{i=0}^L     (^{i-1}_{\,\,j}) \, = \, \sum_{j=0}^{d-2} 
\sum_{m=0}^{L-1}     (^{m}_{\,\,j}) \,  = \sum_{j=0}^{d-2}  (^{\, \, L}_{j+1}) \,.
$$
From this and (\ref{sum1}) we obtain  the inequality
 \begin{equation}\label{sum3}
 P(\Pi) \ge 1 - 2^{-l_\Delta} \, \sum_{j=0}^{d-2}  (^{\, \, L}_{j+1}) \, ,
\end{equation}
which can be used instead of the more complicated right part of (\ref{sum1}).

Considering that the occurrence of the event $ \Pi $ guarantees that  for  the   encoding set $A_0$ constructed  $\, \, d_h(A_0) \ge d$, 
and combining the last inequality and (\ref{sum3}), we obtain the following 
 \begin{theorem}\label{random} Let $L$, $d$  and $\Delta$  be  integers and let $l$ correspond to the VG bound, see (\ref{l}). 
If Algorithm 3 (randomised)  is applied  
 and the number of check symbols is $l+\Delta$
 (that is, the values of the hash functions are chosen randomly from $\{0,1\}^{l+\Delta}$ according to the uniform distribution),    then the probability of the event $\Pi^*$ that  for the  encoding set $A_0$   $\, \, d_h(A_0) \ge d$, satisfies the following inequalities:
 $$
P( \Pi^*)  \ge 1 - 2^{- (l+\Delta)} \, \sum_{j=0}^{d-2}  (^{\, \, L}_{j+1}) \, .
 $$
\end{theorem}
{  \it Corollary.}  
Clearly,  $ \sum_{j=0}^{d-2}  (^{\, \, L}_{j+1}) \,$ $< 2^{l} $.
From this and the theorem   we obtain 
$$ - \log (1 - P(\Pi^*) ) \ge \Delta + O(1) \, ,$$ if $L \to \infty$.

So we can see that the probability of getting an encoding set $ A_0 $ with $ d_h (A_0) \ge d $ is mainly determined by the value of $ \Delta $, that is, the number of extra bits that are added to the VG  bound $l$. This gives a possibility to build simple 
 error detection codes for which the number of extra check bits $\Delta$ does not depend on the length of the message ($L$) and the number of errors that can be detected ($d-1$).

\section{A general method for errors of any type}   
  Now we consider a general case where there is a length of  transmitted (or stored) messages $L$  and a set of possible  distortions $D$ $\subset \{0,1\}^L$ such 
  that any input message $x$ can be received as $x\oplus d,  \, d \in D$.
  For example, let $L=5$ and $D = \{ 00111, 01110, 11100 \} $.  It means that three consecutive letters can be changed. If $x = 01010$ and $d = 11100$,
the output message is $y = 10110$.  

So far, we have considered the case where the check symbols are located at the end of the message.
Now it will be convenient to assume that the check symbols can be located in different positions,
 but, of course, they will be known to the encoder and decoder. This generalization allows us to simplify the notation slightly.

\subsection{Error detection.}    
Let us describe an  algorithm for calculating a hash function $\lambda$ that gives a possibility to detect any distortion $d \in D$, $D \subset \{0,1\}^L \setminus e^L_0$.  That is,
for any message $x$ and any $d \in D$ 
\begin{equation}\label{gen-hasu}
\lambda(x \oplus d) = \lambda( d) \neq 00...0 ; \quad
\lambda(x ) = 00...0 \, \, .
\end{equation}
In order to describe the algorithm we define the sets $D_i$ and $ D'_i$, $i = 1, 2, ... , L$, by 
\begin{equation}\label{Di2}
D_i = \{ d = d_1 ... d_L : d \in D ,\, \, d_i = 1 \,\,   and \, \, d_{i+1} = 0, 
d_{i+2} = 0, ... , d_L = 0 \} \, $$ 
$$ D'_i  = \{ d':  \exists  d \in D_i  \, \, for \, \, which \, \, d'=( d \oplus e^L_i  )  \,\},
\end{equation}
that is, $ D_i$ contains all $ d = d_1 ... d_L $ from $D$ for which 
 $ d_i = 1 \,\,   and \, \, d_{i+1} = 0, 
d_{i+2} = 0, ... , d_L = 0 $,  while $ D'_i $ contains all the words from $ D_i $ in which $ d_i $ changes to 0.

\emph{ Input. }  A message length $L$ and a set of possible distortions $  D \subset \{0,1\}^L \setminus e^L_0$.

\emph{  Output. }
Such an integer  $l$ that  
\begin{equation}\label{l-gen}
2^l - 1 \ge  \max_{i = 1, ... , L} |D'_i |
\end{equation}
 and a linear hash function $ \lambda: \{0,1 \}^L \to \{0,1 \}^l $, for which  (\ref{gen-hasu}) is true.
If $ l $ in (\ref {l-gen}) is not defined or $ l \ge L $, the algorithm stops and answers that the solution does not exist.

\emph{ The algorithm.} 

\emph{ First step. } 
Calculate $l$ in (\ref{l-gen}) and define 
\begin{equation}\label{lyahat-gen} 
 \lambda(e^L_{1} ) \, = \, e^l_1,  \, \,  \lambda(e^L_{2} ) \, = \, e^l_2,  \, ... ,  \, 
 \lambda(e^L_{l} ) \, = \, e^l_l \, .
\end{equation}

\emph{ Second step. } For $i = {l+1},  {l+2}, ... , L$ define $\lambda(e^L_i)$ as follows:
\begin{equation}\label{Di-gen} 
\lambda(e^L_{i})  = v_i \, \, where \, \, v_i \, \,any \,\, word \, \, from \, \, 
\,  \{0,1\}^l    \setminus \,  \lambda(	D'_i )   
 \,. \end{equation}
 From  (\ref{lyahat-gen}) and  (\ref{Di-gen}) we can see that 
 \begin{equation}\label{unicity-gen1}
for \, \, any \, \,  u \in   \lambda(D)  \, , \quad \, \lambda(u) \neq e^l_0 \, \, .
 \end{equation}
Note that 
 $\lambda(e^L_i ) $, $j = 1, 2, ... , i-1$ are defined when $\lambda(D'_i)$ is calculated,    see (\ref{lyahat-gen}) and (\ref{Di-gen}). 

Now we can describe the method for encoding and decoding.
The positions $1, ... , l$ are used for check symbols, while the other $L-l$ are used for information symbols. 
When encoding, the encoder first puts the information symbols into positions $ \{ l+1, ... , L\}$ 
  and  0's into positions  $ 1, ... , l$.
Denote the obtained word $x^*$  and calculate $\lambda(x^*)$ $= w_1 ... w_l$.
Then put letters $w_1 ... w_l$ into the check positions $1, ... , l$  and denote the obtained word by $x = x_1 ... x_L$.  It should be clear that
$\lambda(x) = 00...0$. 
Indeed, 
$\lambda(x) = $ $ \lambda(x^*) $ $ \oplus \lambda(x\oplus x^*)$ 
$= w_1 ... w_l$ $\oplus ((e^l_{1}\times w_1)  \oplus  $  $  (e^l_{2} \times w_2) $ $ \oplus  $ $( e^l_{l}\times w_l) )$ $= w_1 ... w_l \oplus w_1 ... w_l$ $= 00 ... 0\, .$
(Here we used the definition (\ref{lyahat-gen}) ).

It will be convenient to describe the properties of the  algorithm above as follows:

 \begin{theorem}\label{gen-detect}  The  algorithm is correct, that is, 
 if an  error  $d \in D$ has occurred,  then $\lambda ( receved \, \, message) $ $\neq 00...0$, and
 $\lambda ( receved \, \, message) $ $= 00...0$, if no error occurred.
 \end{theorem}
 
 \emph{Proof. } Suppose that the input  
  message is $x = x_1 ... x_L$ and the output message is $y=y_1 ... y_L$.  Then
 $$\lambda(y) = 00...0 \oplus \lambda(y) = \lambda(x) \oplus \lambda(y) =
 \lambda(x\oplus y) \in D \, .
 $$
 Taking into account (\ref{unicity-gen1}), from these equations we can see that 
 $\lambda(y) = 00...0$ if $y = x$ and $\lambda(y) \neq 00...0$, if $y \neq x$.

 Now consider the complexity of the proposed method.
   There are two important characteristics: the time of encoding and decoding  and the construction time of the  hash function. 
   It is important to note that the hash function must be prepared once, and then can be used for a long time, while encoding and decoding are performed repeatedly.

 \begin{claim}\label{detct}   The number of bit-operations ($t$)  for encoding and decoding is not grater than
$
O(L\, l)$,  If  $ L $ grows    to $ \infty $. 
The  number of bit-operations ($T$)  for building the hash-function $\lambda$   is proportional to $|D| \, l$.  
 \end{claim} 
 
 \emph{Proof } is based on a direct estimation of the number of bit-operations.
 
Let us consider a simple 
 example illustrating the described method.
Suppose that a system should transmit 6-bit messages, but two consecutive letters may be distorted.  It means that the set of possible distortions is 
\begin{equation}\label{ex1}
D = \{ 000011, 000110, 001100, 011000, 110000 \} \, .
\end{equation}
(That is, any message $x_1 ... x_6$ may  change into $x\oplus d_i$ 
during the transmission, where $d_i$ is $i$-th word from $D$.)
Our goal is to build a code which can detect any distortion from $D$ that  occurs during the transmission.  For this, we first build a linear hash-function  $\lambda$ described in this part. 
According to (\ref{Di2}) we find that $D_1$ and $D'_1$  are empty sets and 
$$ D_2  = \{110000\}, \, D_3 =  \{011000\}, \, D_4 = \{001100\}, \,D_5=\{000110\}, \, D_6 =\{000011\} \, ,
$$
$$ D'_2  = \{100000\}, \, D'_3 =  \{010000\}, \, D'_4 = \{001000\}, \,D'_5=\{000100\}, \, D'_6 =\{000010\} \, .
$$
Clearly, $\max_i |D'_i| = 1$ and from (\ref{l-gen}) we obtain that 
$2^l - 1 \ge 1$ and, hence, it is enough to put $l = 1$. Recall that it means that there will be one check symbol and 5 information ones, and, besides, $\lambda$ will take values from 
$\{0,1\}$.

Now we can find  $\lambda$. According to (\ref{lyahat-gen})
we obtain $\lambda(e^6_1 )=  1$.  Then, based on (\ref{Di-gen}) we calculate
all the rest of the  values of $\lambda$ as follows:
$\lambda(e^6_2)$ should be chosen from the set $\{0,1\} \setminus \{1\} = \{0\}$.
So, $\lambda(e^6_2) = 0$. Analogously, 
 $\lambda(e^6_3) = 1$, $\lambda(e^6_4) = 0$,
$\lambda(e^6_5) = 1$, $\lambda(e^6_6) = 0$. Or, to put it shortly,
$\lambda(e^6_{even}) = 0$, $\lambda(e^6_{odd}) = 1$.
 From (\ref{ex1}), we can see  that $\lambda(d) = 1$ for any distortion $d$ $\in D$ 
and, hence, any  distortion from this set  is detected. 

Now we can finish the description of the code. We know that $l=1$ and, hence, the first message symbol $x_1$ is a check symbol,   while $x_2 ... x_6$ are information ones.
Suppose that information symbols are $11001$. The encoder forms the word 
$x^* = 011001$, calculates $\lambda(x^*) = 1$ and, hence, $x = 1 11001$. 
If no error  occurs, then $\lambda(x) = 0$ and the receiver obtains the information symbols
$11001$. If a distortion $d$ occurs (say, $d = 011000$), the receiver obtains the word
$y = x \oplus d = 100001$, calculates $ \lambda(100001) = 1$ and sees that the message  was corrupted during the transmission.

 \subsection{Error-correction.}
In this part we  describe an  algorithm for calculating a hash function $\lambda$ that gives a possibility to correct  any distortion $d \in D$, where $D$  is a given subset 
from $\{0,1\}^L$.  
We say that the system corrects distortions from $D$ if 
\begin{equation}\label{i1}
\lambda (x) = 00...0 \, \,  for  \, \,  any  \, \,  input  \, \,  message  \, \,  x  \, , \, \,\, 
all  \, \, \lambda(d),\, \,  d \in D  \, \,  are  \, \,  different  \, \, 
$$ $$ and  \, \, non-equal  \, \, to  \, \,  00...0  \, ,\, (hence, \, \,  \lambda (x\oplus d) = \lambda(d) \neq 00...0 \, )\, . 
\end{equation}
Note that this property  gives a possibility to find  $d$ and the original 
 message $x = y \oplus d$.

To describe the algorithm for constructing $ \lambda $, we will define some auxiliary variables. For any word $x_1 ... x_L$ and $1 \le i  \le L $  we  define $x|_1^i = x_1 x_2 ... x_i 00 ... 0$ and let 
\begin{equation}\label{GHF} 
D^+ = D \cup \{ e^L_0 \}, \, G_i = \{d|_1^i, d \in D^+ \},    H_i = G_i \setminus G_{i-1},  \, \, 
$$ $$
F_i = \{ All \, \, f \,\,for \,\,which \,\, \exists  g \in G_{i-1}, \exists h \in H_i  :
\, \, f = g \oplus h \oplus e^L_i   \} \,.
\end{equation}

\emph{ Input. }  A message length $L$ and a set of possible distortions $  D \subset ( \{0,1\}^L \setminus e^L_0 ) $.

\emph{  Output. }
An integer  $l$   such that
\begin{equation}\label{l-gen-corr}
2^l - 1 \ge  \max_{i = 1, ... , L} \, \,    \, |F_{i} | \, 
\end{equation}
 and a linear hash function $ \lambda: \{0,1 \}^L \to \{0,1 \}^l $, for which  
\begin{equation}\label{i1-lam}
all  \, \, \lambda(d),\, \,  d \in D,  \, \,  are  \, \,  different  \, \, 
 and  \, \, non-equal  \, \, to  \, \,  00...0  \, ,
\end{equation}
see (\ref{i1}).
 If $ l $ in (\ref {l-gen-corr}) is not defined or $ l \ge L $, the algorithm stops and answers that the solution does not exist.

\emph{ The algorithm.} 

\emph{ First step}.
Define 
\begin{equation}\label{E22} 
\lambda(e^L_{1} ) \, = \, e^l_1,  \, \,  \lambda(e^L_{2} ) \, = \, e^l_2,  \, ... ,  \, 
 \lambda(e^L_{l} ) \, = \, e^l_l \, .
\end{equation}

\emph{ Second step. } For $i = l+1,  l+2, ... , L$ define 
\begin{equation}\label{E2} 
\lambda(e^L_{i})  = v \, \, where \, \, v \, \,any \,\, word \, \, from \, \, 
  \{0,1\}^l    \setminus \,  \lambda(	F_i ) \, 
   . \end{equation}
Note that 
 $\lambda(e^L_j ) $, $j = 1, 2, ... , i-1$ are defined when $\lambda(F_i)$ is calculated,     
see (\ref{E22}) and (\ref{E2}). 

The key property of the  algorithm described is the following 
\begin{theorem}\label{ghf}
 \begin{equation}\label{unicity-gen}
All \, \, \lambda(u), \,  u \in  \lambda(D^+), \text{ are different.}
 \end{equation}
 \end{theorem}
 {\it Proof.} 
We prove this by induction on $i$ for $G_i, i = 1, ... , L$,  where $G_L = D^+$. For $1, ... , l$ the property (\ref{unicity-gen}) follows from (\ref{E22}), because    $x_{{1}}  x_{{2}} ... x_{{l}}$ $= \lambda(
x|_1^{l})$
for any  $x = x_1 ... x_L$. Suppose that (\ref{unicity-gen}) is proven for $G_i$, and let us prove it for $G_{i+1}$. Let $u, v \in G_{i+1}$. We  need to show that $\lambda(u) \neq \lambda(v)$. 
There are the  following  three possibilities: 

i) $u, v \in G_i$. Then, $\lambda(u) \neq \lambda(v)$, because it is proven for $G_i$.

ii) $u, v \in G_{i+1} \setminus G_i \, \, ( = H_{i+1})$.  In this case
$u \oplus e^L_{i+1} \in G_i$ and $v \oplus e^L_{i+1} \in G_i$ (i.e. both belong to $G_i$)  and, hence,  $\lambda(u \oplus e^L_{i+1} ) \neq $ $\lambda(v \oplus e^L_{i+1})$.  So,
$\lambda(u ) \neq $ $\lambda(v )$.

iii) $u \in G_i, v \in H_{i+1}$.  In this case $\lambda (u) \oplus \lambda(v) $
$= \lambda(u \oplus v \oplus e^L_{i+1}) \oplus \lambda(e^L_{i+1}) .$
From the definition (\ref{GHF})  we can see that $u \oplus  v \oplus e^L_{i+1}$ belong to $F_{i+1}$. Taking into account (\ref{E2}), we can see  that 
$ \lambda( e^L_{i+1} ) $ $\neq \lambda( u \oplus  v \oplus e^L_{i+1}) $. Hence, $ \lambda(u \oplus v) \neq 00...0 , $ and $\lambda(u) \neq \lambda(v)$.

So, for i) - iii)  the inequality   $\lambda(u) \neq \lambda(u)$ is proven and the induction step is completed.  
 The claim (\ref{unicity-gen}) is proven.

Now we can describe the methods for encoding and decoding.
The encoding coincides with the method for error-detection. Namely, 
the positions $1, ... , l$ are  used for check symbols, while the other $L-l$ are used for information symbols. 
The encoder first puts the information symbols into positions $ \{  1+1, ... , L\}$ 
  and  0's into positions  $ 1, ... , l$.
Denote the obtained word $x^*$  and calculate $\lambda(x^*)$ $= w_1 ... w_l$.
Then put letters $w_1 ... w_l$ into the check positions $1, ... , l$  and denote the obtained word by $x = x_1 ... x_L$.  It should be clear that
$\lambda(x) = 00...0$. 
Indeed, 
$\lambda(x) = $ $ \lambda(x^*) $ $ \oplus \lambda(x\oplus x^*)$ 
$= w_1 ... w_l$ $\oplus ((e^l_{1}\times w_1)  \oplus  $  $  (e^l_{2} \times w_2) $ $ \oplus  $ $( e^l_{l}\times w_l) )$ $= w_1 ... w_l \oplus w_1 ... w_l$ $= 00 ... 0\, .$
(Here we used the definition (\ref{E22}).)
The decoder calculates $\lambda(  y)$ for the received (or stored) $y$.
If $\lambda(  y) = 00...0$, then no error  occurred, otherwise a distortion $d$ has occurred,
for which $\lambda(  d) $ $= \lambda(  y)$ (and, hence $y \oplus d$ is the original message).

From the property (\ref{unicity-gen}) we can see that the described method is correct.
 
Let us consider the complexity of the proposed method. 
   There are three important characteristics: the encoding  time ($t_{enc}$),  decoding 
one   ($t_{dec}$)  and the construction time of the  hash function in accordance with the described algorithm ($T$).
It is clear that $t_{enc}$ $= O(L \log L)$ and $t_{dec} = O(|D| L \log L)$.
Basing on (\ref{GHF}) and (\ref{E2}) we can obtain an estimate $T=$ $O(|D| \log L )^3$.

    Let us consider an example. Let the set of possible distortions  $D$  be (\ref{ex1}). Then, from (\ref{GHF}) we obtain
 $$
 D^+ = \{ 000000, 000011, 000110, 001100, 011000, 110000 \} \, ,
$$
$$ G_1 = \{   000000, 100000\},  G_2 = \{   000000, 010000, 110000\}, $$ $$ 
G_3=  \{   000000,  001000,    011000, 110000 \},  
 G_4 = \{   000000,  000100,  001100,  011000, 110000 \},   $$ $$ 
 G_5 = \{  000000,  000110,  011000,  001100,  110000,  000010 \}, G_6 = D^+,
  $$ 
  $$ H_1 = G_1 = \{   000000, 100000\}, H_2= \{   010000, 110000\},
  H_3 = \{     001000,    011000 \},
  $$
$$
H_4=  \{  000100,  001100  \}, H_5 = \{  000010,  000110 \}, H_6 = \{  000011 \},
$$
$$
F_1 = \{ 000000, 100000 \},  F_2 = \{ 000000, 100000 \}, F_3 = \{ 000000, 100000,
010000, 110000  \},
$$
$$
 F_4 = \{ 000000, 010000, 110000, 011000, 001000, 111000   \},
$$
$$
 F_5 = \{ 000000,  000100, 001100, 001000, 011000, 011100, 110000, 110100   \},
$$
$$
 F_6 = \{ 000000, 000010,  000100,  001110, 011010, 110010   \}.
$$
According to (\ref{l-gen-corr})  we find $\max_{i = 1, ... , L} \, \,  |F_{i} | $
$= |F_5| = 8 $  and, hence, $2^4 - 1 \ge 8$, $\, \, l=4 .$ 
From (\ref{E22}) we obtain
$$\lambda(e^6_1) = 1000, \,  \lambda(e^6_2) = 0100, \,
\lambda(e^6_3) = 0010, \,  \lambda(e^6_4) = 0001 \, . 
$$   
   Then, according to (\ref{E2}) we calculate 
   $$ \{0,1\}^4 \setminus \lambda (F_5) = \{0,1\}^4 \setminus 
   \{ 0000, 0001, 0011, 0010, 0110, 0111, 1100, 1101 \} = $$
   $$
  \{  0100, 0011,  1000, 1001, 1011, 1110, 1111 \}.
      $$
  Any of these words can be chosen as the value of $\lambda(e^6_5)$.    
  So,  let $\lambda(e^6_5)= 1111$. Analogously, 
  $$ \{0,1\}^4 \setminus \lambda (F_6) = \{0,1\}^4 \setminus 
   \{ 0000, 1111,  0001, 1100, 1001, 0011 \} = 
   $$ $$\{ 0010, 0011, 0100, 0101, 0110, 0111,
   1000, 1010, 1101, 1110 \} .
   $$
   Thus, we can define $\lambda(e^6_6)= 1000$.
   (We can check that all $\lambda(d), d \in D, $ are different:  $\lambda (000011) = 
  0111,  $ $\lambda (000110) = 
  1110,  $ $\lambda (001100) = 
  0011,  $ $\lambda (011000) = 
  0110,  $ $\lambda (110000) = 
  1100 .$)
   So, a linear hash function has been constructed, the number of information symbols is
   $L-l = 6-4=2$, the number of check symbols is $l=4$. 
   Suppose that the information symbols are 10. Then, according to the encoding method,
   $x^* = 000010, \lambda(x^*) = 1111$, $ x = 111110$. Suppose that the distortion 
   $011000$ has occurred. Then $y = 100110$, $\lambda(y) = 0110$. 
   Note that $\lambda(011000) = 0110$. Thus, $0110 = \lambda(y) = \lambda(011000) $. 
   It means that the decoder has found the distortion $d =  (011000)$ and can find
   $x = y \oplus d$ $= 100110 \oplus 011000 = 111110 = x$.
   So, the error is corrected.

\section{Conclusion}
In this paper we have shown how linear hash functions can be used for error detection and error correction.  It turns out, that it is possible to build error detection and correction  codes for any possible set of distortions. 

The case when the number of errors does not exceed a predetermined value is discussed in more detail. We consider a method whose performance is slightly better than the Varshamov - Gilbert   
 bound \cite{vp}. In addition, we propose a randomized algorithm, the performance of which is close to this bound, but the construction and encoding times are close to linear.

\section*{Acknowledgment}
This work was supported by Russian Foundation for Basic Research (grant 18-29-03005).

\end{document}